\documentclass[11pt]{article}

\newif\ifnotes\notesfalse
\newif\ifsub\subfalse

\title{Asymptotically Good Quantum Codes with \\ Addressable and Transversal Non-Clifford Gates

}

\ifsub
\author{\textit{Anonymous Submission to FOCS 2025}}
\else 
\author{Zhiyang He\thanks{Email: \texttt{szhe@mit.edu}. Supported by the MIT Department of Mathematics and the NSF Graduate Research Fellowship Program under Grant No. 2141064.}\\MIT
\and Vinod Vaikuntanathan\thanks{Email: \texttt{vinodv@mit.edu}. Supported in part by NSF grant CNS-2154149 and a Simons Investigator Award.}\\MIT
\and Adam Wills\thanks{Email: \texttt{a\_wills@mit.edu}. Affiliated with the MIT Center for Theoretical Physics - a Leinweber Institute. Supported by the MIT Department of Physics. This preprint is assigned number MIT-CTP/5859.}\\MIT
\and Rachel Yun Zhang\thanks{Email: \texttt{rachelyz@mit.edu}. Supported by NSF Graduate Research Fellowship 2141064. Supported in part by NSF grant CNS-2154149.}\\MIT}
\fi 

\date{\today}

\usepackage[breaklinks]{hyperref}
\usepackage[margin=1in]{geometry}
\usepackage{mathptmx}
\usepackage{amssymb}
\usepackage{amsmath}
\usepackage{amsthm}
\usepackage{amsfonts}
\usepackage[T1]{fontenc}
\usepackage{bm}
\usepackage{color}
\usepackage{comment}
\usepackage[capitalize]{cleveref}
\usepackage[dvipsnames]{xcolor}
\usepackage{float}
\usepackage{todonotes}
\usepackage{asymptote}
\usepackage{mdframed}
\usepackage[most]{tcolorbox}
\usepackage{hyperref}
\usepackage[shortlabels]{enumitem}
\usepackage{framed}
\usepackage{mdframed}
\usepackage{scrextend}
\usepackage{multirow}
\usepackage{ifthen}
\usepackage{bbm}
\usepackage{frcursive}
\usepackage{thm-restate}
\usepackage{booktabs}
\usepackage{array}
\usepackage{graphicx}
\usepackage{physics}
\usepackage{mathtools}

\newcommand{\rnote}[1]{\ifnotes $\ll$\textsf{\color{red} Rachel: { #1}}$\gg$ \fi}
\newcommand{\snote}[1]{\ifnotes $\ll$\textsf{\color{blue} Sunny: { #1}}$\gg$ \fi}

\definecolor{denim}{rgb}{0.08, 0.38, 0.74}
\definecolor{classicrose}{rgb}{0.98, 0.8, 0.91}
\definecolor{darkpastelblue}{rgb}{0.47, 0.62, 0.8}
\definecolor{dogwoodrose}{rgb}{0.84, 0.09, 0.41}

\usepackage{hyperref}
\hypersetup{
    colorlinks=true,
    linkcolor=dogwoodrose,
    filecolor=dogwoodrose,
    citecolor=denim,
    urlcolor=denim,
}
\usepackage[hyperpageref]{backref}

\newtheorem{theorem}{Theorem}[section]
\newtheorem{lemma}[theorem]{Lemma}
\newtheorem*{lemma*}{Lemma}
\newtheorem{claim}[theorem]{Claim}

\newtheorem{proposition}[theorem]{Proposition}

\newtheorem{assumption}[theorem]{Assumption}
\newtheorem{definition}[theorem]{Definition}

\theoremstyle{definition}
\newtheorem{remark}[theorem]{Remark}

\Crefname{theorem}{Theorem}{Theorems}
\Crefname{claim}{Claim}{Claims}
\Crefname{lemma}{Lemma}{Lemmas}
\Crefname{proposition}{Proposition}{Propositions}
\Crefname{corollary}{Corollary}{Corollaries}
\Crefname{definition}{Definition}{Definitions}

\renewcommand{\leq}{\leqslant}
\renewcommand{\ge}{\geqslant}
\renewcommand{\geq}{\geqslant}

\newcommand{\CCZ}{\mathsf{CCZ}}

\newcommand{\CSS}{\text{CSS}}

\newcommand{\AG}{\mathsf{AG}}

\newcommand{\bbN}{\mathbb{N}}

\newcommand{\bbF}{\mathbb{F}}

\newcommand{\cC}{\mathcal{C}}
\newcommand{\cD}{\mathcal{D}}

\newcommand{\cL}{\mathcal{L}}

\newcommand{\cN}{\mathcal{N}}

\newcommand{\cQ}{\mathcal{Q}}

\newcommand{\cT}{\mathcal{T}}

\newcommand{\rr}{r}
\newcommand{\Gal}{\text{Gal}}

\makeatletter
\newcommand{\customlabel}[2]{%
   \protected@write \@auxout {}{\string \newlabel {#1}{{#2}{\thepage}{#2}{#1}{}} }%
   \hypertarget{#1}{#2}
}
\makeatother

\newcounter{datacounter}

\newcounter{casenum}

\newcommand{\case}[2]{
    \refstepcounter{casenum}
    \ifthenelse{\equal{\value{casenum}}{0}}{
    \vskip.5\baselineskip\par\noindent
    }{}
    \noindent {\it Case \arabic{casenum}:} {\it #1}
    \vskip0.1\baselineskip
    \begin{addmargin}[1.5em]{1em}
    #2
    \end{addmargin}
}

\newcounter{subcasenum}

\newcounter{casenumb}

\newcounter{subcasenumb}

\usepackage{cite}

\newcommand{\bal}{\boldsymbol{\alpha}}

\newcommand{\balpha}{\boldsymbol{\alpha}}
\newcommand{\bbeta}{\boldsymbol{\beta}}

\newcommand{\qcode}{\mathcal{Q}}

\usepackage[normalem]{ulem}

\newcommand{\FF}{\mathbb{F}}

\newcommand*{\ol}{\overline}  
\usepackage{braket}
\usepackage[normalem]{ulem}
\DeclareMathAlphabet{\mathcal}{OMS}{cmsy}{m}{n}

\DeclareMathOperator{\Aut}{Aut}

\begin{document}

\sloppy
\maketitle
\ifsub
\else 
\renewcommand{\thefootnote}{\fnsymbol{footnote}}
\setcounter{footnote}{5}
\newcommand\blfootnote[1]{
    \begingroup
    \renewcommand\thefootnote{}\footnote{#1}
    \addtocounter{footnote}{-1}
    \endgroup
}
\blfootnote{$^{\P}$The authors are listed in alphabetical order.}
\setcounter{footnote}{0}
\renewcommand{\thefootnote}{\arabic{footnote}}
\fi 

\begin{abstract}
\snote{alternative titles:
asymptotically good quantum codes with addressable and transversal non-Clifford gates.}

Constructing quantum codes with good parameters and useful transversal gates is a central problem in quantum error correction. In this paper, we continue our work in~\cite{he2025quantum} and construct the first family of asymptotically good quantum codes (over qubits) supporting transversally addressable non-Clifford gates. More precisely, given any three logical qubits across one, two, or three codeblocks, the logical $\mathsf{CCZ}$ gate can be executed on those three logical qubits via a depth-one physical circuit of $\mathsf{CCZ}$ gates.
This construction is based on the transitive, iso-orthogonal algebraic geometry codes constructed by Stichtenoth~\cite{stichtenoth2006transitive}. 
This improves upon our construction from~\cite{he2025quantum}, which also supports transversally addressable $\mathsf{CCZ}$ gates and has inverse-polylogarithmic rate and relative distance.

\end{abstract}
\thispagestyle{empty}
\newpage

\newgeometry{left=2cm,right = 2cm,top=1.75cm,bottom=1.75cm}
\tableofcontents
\pagenumbering{roman}
\restoregeometry
\newpage
\pagenumbering{arabic}

\section{Introduction}


Since the discovery and formulation of quantum error correction~\cite{shor1995scheme,steane1996multiple,calderbank1996good,gottesman1997stabilizer}, quantum codes with good parameters and useful transveral gates has been a topic of major interest in the field. In this direction, constructing \textit{high rate} quantum codes with transversal and \textit{addressable} quantum gates is a particularly difficult problem. The former property ensures a very space-efficient error correction overhead, whereas the latter property enables a particularly fault-tolerant and time-efficient implementation of quantum logic.
We will omit further discussion on this property and its motivation here, referring the reader to Section~1 of~\cite{he2025quantum} for a complete discussion. 
In this paper, which is a follow-up to~\cite{he2025quantum}, we construct the first family of asymptotically good quantum codes with transversal and addressable non-Clifford gates, as detailed in the following theorem.


\begin{theorem}\label{thm:main_result_1}
    There exists a family of asymptotically good quantum CSS codes over qubits with parameters
    \begin{equation}
        \left[\left[n,\Theta(n), \Theta(n)\right]\right]_2,
    \end{equation}
    supporting a transversally addressable non-Clifford gate. Specifically, any three logical qubits across one, two, or three blocks of the code may be addressed with a logical $\CCZ$ gate via a depth-one circuit of physical $\CCZ$ gates. 
\end{theorem}

This construction follows similar ideas to the construction of Theorem~1.1 in~\cite{he2025quantum}. 
In that work, we used punctured Reed-Solomon codes as the component classical codes in a CSS construction.  
The multiplication property of the Reed-Solomon code (see Fact~2.11 in~\cite{he2025quantum}) and its automorphisms by affine shifts enable the tranversal implementations of addressable non-Clifford gates. 
A caveat of this construction, however, is that Reed-Solomon codes are defined over finite fields $\mathbb{F}_q$ of growing order, which means the constructed CSS codes are defined over qudits of increasing dimension. 
In particular, the dimension $q$ must be at least $n$, the number of physical qudits in the code. 
With established techniques (see Section~6 of~\cite{he2025quantum} and also~\cite{nguyen2024goodbinaryquantumcodes,golowich2024asymptoticallygoodquantumcodes}), we can perform a code concatenation procedure to construct a qubit code with the same transversal addressability of the $\CCZ$ gate.
However, this concatenation incurs a loss of parameters: while our codes are asymptotically good over qudits of growing dimension (see Theorem~3.1 of~\cite{he2025quantum}), the resulting qubit code loses a factor \(\text{polylog}(q)\) during the concatenation procedure (see Theorem~1.1 of~\cite{he2025quantum}). Therefore, since $q \sim n$, this means that the resulting qubit code had parameters
\begin{equation}
    \left[\left[ n, \Theta\left(\frac{n}{\text{polylog}(n)}\right), \Theta\left(\frac{n}{\text{polylog}(n)}\right)\right]\right]_2.
\end{equation}

In this follow-up work, we use classical algebraic geometry codes to build quantum codes with the same properties, except where the qudit dimension is constant; $q = \Theta(1)$. This means that the resulting qubit codes after the concatenation procedure is also asymptotically good, thus establishing Theorem~\ref{thm:main_result_1}.
In particular, we employ the transitive, iso-orthogonal algebraic geometry codes due to Stichtenoth~\cite{stichtenoth2006transitive}, where the transitive action of the Galois automorphisms on the rational places of the codes enables the addressability of non-Clifford gates.
This answers one of the open problems from the discussion in Section~1.2 of~\cite{he2025quantum}. 
We remark that the other open problems in the list remain relevant and interesting.

\section{Preliminaries}\label{sec:prelim}
For the sake of brevity, we will mostly refer the reader to existing preliminary material. In particular, for the background on finite fields, non-Clifford gates and algebraic codes, we refer the reader to the preliminary material in~\cite{he2025quantum}. However, we state the most essential information here for convenience.

Consider $\mathbb{F}_q$, the finite field of order $q$, where $q$ is a power of two.\footnote{In this work, we are interested in the case when $q$ is a power of two because we ultimately want codes over qubits, although the results may be easily generalised to qudits of any dimension.} A (Galois) qudit of dimension $q$ has computational basis states denoted $\left\{\ket{\eta}\right\}_{\eta \in \mathbb{F}_q}$. For each $\gamma \in \mathbb{F}_q$, there is a $\mathsf{CCZ}$ gate, denoted $\mathsf{CCZ}^\gamma$, acting on three qudits as
\begin{equation}
    \mathsf{CCZ}^\gamma\ket{\eta_1}\ket{\eta_2}\ket{\eta_3} = (-1)^{\tr(\gamma\;\eta_1\eta_2\eta_3)}\ket{\eta_1}\ket{\eta_2}\ket{\eta_3},
\end{equation}
where $\tr:\mathbb{F}_q \to \mathbb{F}_2$ is an $\mathbb{F}_2$-linear map known as the trace~\cite{mullen2013handbook}. We use $\overline{\mathsf{CCZ}^\gamma}$ to denote the logical $\mathsf{CCZ}^\gamma$ gate acting on three logical qudits, where the code will be obvious from context. In such a case, the notation $\overline{\mathsf{CCZ}^\gamma[A,B,C]}$ denotes the logical gate $\mathsf{CCZ}^\gamma$ acting on three logical qudits denoted $A$, $B$ and $C$. We may also wish to consider acting with the logical gate on qudits in different blocks of the same code (a so-called ``inter-block gate''). Given codeblocks labelled $1$, $2$ and $3$, and qudits in the respective codeblocks denoted $A, B$ and $C$, the logical $\mathsf{CCZ}^\gamma$ gate acting on these three qudits is denoted $\overline{\mathsf{CCZ}_{123}^\gamma[A,B,C]}$.

In the present work, we will use some properties of specialised algebraic geometry codes. For the initial material on algebraic geometry codes, we refer the reader to Section~2.3 of~\cite{wills2024constant}. The present work requires more in-depth material than this, though, and we state the necessary further facts in the following section.

\subsection{Further Preliminaries on Algebraic Geometry Codes}\label{sec:further_AG_prelims}

Our construction will rely heavily on the work~\cite{stichtenoth2006transitive} due to Stichtenoth. As is common in such constructions, function fields with desirable properties, including having many rational places, are obtained by defining \textit{towers} of function fields. A \textit{tower} of function fields is a sequence $\mathcal{E} = (E_0 \subseteq E_1 \subseteq E_2 \subseteq \ldots )$ where every $E_i$ is an algebraic function field of one variable over $\mathbb{F}_q$, and each extension $E_{i+1}/E_i$ is a field extension of finite degree.
For our purposes, we will be concerned with the tower $\mathcal{E}$ constructed in~\cite{stichtenoth2006transitive}. Here, each $E_i$ is a function field defined over $\mathbb{F}_q$, where $q=r^2$ is a prime power. In addition, for the purposes of this paper, we let $q$ be a power of two, because we ultimately want to construct codes over qubits via the qudit-to-qubit transformations described in~\cite{he2025quantum} (although this may readily be generalised to qudits of any dimension). 

Now, it will be important to be able to describe the places of $E_i$ in terms of places of function fields lower down the tower. For this, we give the following important definition.
\begin{definition}[From Definition 3.1.3 of~\cite{stichtenoth2009algebraic}]\label{def:place_extension}
    Consider function fields $F$ and $F'$ over $\mathbb{F}_q$, where $F \subseteq F'$ is a field extension of finite degree. Given a place $P$ of $F$, and another place $P'$ of $F'$, we say that $P'$ lies over $P$ if $P' \supseteq P$. We also that $P'$ is an extension of $P$ or that $P$ lies under $P'$, and denote the situation as $P'|P$.
\end{definition}
It then becomes natural to ask, in the notation of Definition~\ref{def:place_extension}, if we are given a place $P$ of $F$, how many places $P'$ of $F'$ may lie over $P$? An upper bound follows.
\begin{proposition}[From Corollary 3.1.12 of~\cite{stichtenoth2009algebraic}]
    With the notation of Definition~\ref{def:place_extension}, we have that the number of places $P'$ of $F'$ that may lie over $P$ is at most the degree of the field extension, $[F':F]$. 
\end{proposition}
Given the above, the following definition is very natural.
\begin{definition}[From Definition 3.1.13 of~\cite{stichtenoth2009algebraic}]
    With the notation of Definition~\ref{def:place_extension}, we say that the place $P$ splits completely in the extension $F'/F$ if there are exactly $[F':F]$ distinct places $P'$ of $F'$ that lie over $P$.
\end{definition}
Now, the transversal addressability of the $\CCZ$ gates in our quantum codes will rely on the transitivity of the codes in~\cite{stichtenoth2006transitive} (see Theorem~\ref{thm:Galois_transitivity}). To present these, we give the following definitions, which are standard concepts in abstract algebra, in particular in Galois theory.
\begin{definition}
    Consider a field extension $M/L$ of finite degree. The group of automorphisms of the extension is the group of field automorphisms of $M$ that fix $L$ pointwise, that is,
    \begin{equation}
        \Aut(M/L) \coloneq \{\sigma:M \to M\mid \sigma \text{ is a field isomorphism and }\sigma(a) = a \text{ for all } a \in L\}.
    \end{equation}
    It turns out that $|\Aut(M/L)| \leq [M:L]$. In the case that $|\Aut(M/L)| = [M:L]$, the extension $M/L$ is known as a \textit{Galois extension}. In this case, we denote $\Gal(M/L) \coloneq \Aut(M/L)$, and call it the \textit{Galois group} of the extension.
\end{definition}
Galois extensions will be particularly important for us, as will start to become clear shortly. To set ourselves up for this, we consider the following lemma.
\begin{lemma}[From Lemma 3.5.2 of~\cite{stichtenoth2009algebraic}]\label{lemma:place_group_action}
    Given function fields $F'$ and $F$ over $\mathbb{F}_q$, for which $F \subseteq F'$, let $P$ be a place of $F$ and $P'$ a place of $F'$. Considering some $\sigma \in \Aut(F'/F)$, we have that
    \begin{equation}
        \sigma(P') \coloneq \{\sigma(z): z \in P'\}
    \end{equation}
    is a place of $F'$. The valuation of a function $f \in F'$ at the place $\sigma(P')$ is given by
    \begin{equation}
        \nu_{\sigma(P')}(f) = \nu_{P'}(\sigma^{-1}(f)).
    \end{equation}Moreover, if $P'$ lies over $P$, then $\sigma(P')$ also lies over $P$.
\end{lemma}
This lemma is telling us that the automorphism group of the field extension induces a group action on the set of places of $F'$ that lie over a given place of $F$. The importance of Galois extensions is seen in the following theorem, which tells us that if the extension $F'/F$ is Galois, this action is in fact transitive.
\begin{theorem}[From Theorem 3.7.1 of~\cite{stichtenoth2009algebraic}]\label{thm:Galois_transitivity}
    Consider function fields $F$ and $F'$ over $\mathbb{F}_q$, where $F'/F$ is a Galois extension. Further, consider places $P_1$ and $P_2$ of $F'$ which both lie over a place $P$ of $F$. Then, there exists some $\sigma \in \Gal(F'/F)$ for which $P_2 = \sigma(P_1)$. Alternatively stated, the Galois group of the extension acts transitively on the set of extensions of the place $P$.
\end{theorem}
The tower of function fields $\mathcal{E}$ constructed by Stichtenoth in~\cite{stichtenoth2006transitive} has the special property that all extensions $E_i/E_0$ are Galois extensions (see Theorem 1.7 in that paper). The function field has some further useful structure, which may be presented by writing the following (see Section 2 in~\cite{stichtenoth2006transitive}):
\begin{equation}
    E_0 = \mathbb{F}_q(z) \subseteq \mathbb{F}_q(w) \subseteq F_0 \coloneq \mathbb{F}_q(x_0) \subseteq E_1 \subseteq E_2 \subseteq \ldots \text{, where } z = w^{r-1} \text{ and } w = x_0^r+x_0.
\end{equation}
We find that $E_0$ is exactly the rational function field over $\mathbb{F}_q$ in the variable $z$. However, between $E_0$ and the next level of the tower, $E_1$, there are two additional rational function fields in the variables $w$ and $x_0$, where $z=w^{r-1}$, and $w = x_0^r+x_0$. For convenience, we have also defined $F_0 = \mathbb{F}_q(x_0)$.

Stichtenoth builds algebraic geometry codes using places of function fields $E_i$ that lie over a particular place of $E_0$, and we will do something similar. Therefore, it will be important to have some information about the places of rational function fields.
\begin{theorem}[From Section 1.2 of~\cite{stichtenoth2009algebraic}]
    Let $\mathbb{F}_q(x)$ be the rational function field in the variable $x$ over $\mathbb{F}_q$. The rational places of $\mathbb{F}_q(x)$ are in natural one-to-one correspondence with the set $\mathbb{F}_q\cup\{\infty\}$. In particular, this function field has exactly $q+1$ rational places.
\end{theorem}
Following~\cite{stichtenoth2006transitive}, we employ the following notation for the rational places of a rational function field $\mathbb{F}_q(x)$.
\begin{definition}
    The $q+1$ rational places of $\mathbb{F}_q(x)$ are denoted $(x=\alpha)$, where $\alpha \in \mathbb{F}_q$, and $(x=\infty)$.
\end{definition}
We are now in a position to describe Stichtenoth's construction of transitive codes in~\cite{stichtenoth2006transitive}. Referring to Theorem 1.7 in that paper, he considers the rational place $(z=1)$ of $E_0 = \mathbb{F}_q(z)$, and shows that this place splits completely in the extension $E_i/E_0$ for any $i \in \mathbb{N}$. He also shows that the degree of each extension is
\begin{equation}
    N^{(i)}\coloneq [E_i:E_0] = (r-1)r^i2^{t(i)},
\end{equation}
where $t(i)$ is a non-negative integer,\footnote{Note that in the statement of Theorem 1.7 of~\cite{stichtenoth2006transitive}, $p$ denotes the character of the field $\mathbb{F}_q$, but we are only interested in finite fields with characteristic $p=2$, leading to our statement.} which means that $E_i$ has exactly $(r-1)r^i2^{t(i)}$ places lying over the place $(z=1)$ of $E_0$ and, moreover, these places are all rational.\footnote{This latter statement is technically redundant; it turns out that if a rational place $P$ splits completely in a field extension of finite degree that all the places lying over $P$ must themselves be rational; see Section 3.1 of~\cite{stichtenoth2006transitive}.} Stichtenoth then constructs algebraic geometry codes by evaluating the functions in a Riemann-Roch space at these rational places that lie over the place $(z=1)$ of $E_0$. Because the extension $E_i/E_0$ is Galois, the group $\Gal(E_i/E_0)$ acts transitively on this set of places. Stichtenoth used this fact to construct classical codes with asymptotically good parameters which are \textit{transitive}, meaning that the group of automorphisms of the code acts transitively on the code's coordinates.\footnote{Given a code $C \subseteq \mathbb{F}_q^N$, a permutation $\sigma \in S_N$ is said to be an automorphism of the code if $(c_1, \ldots, c_N) \in C \implies (c_{\sigma(1)}, \ldots, c_{\sigma(N)}) \in C$. The set of automorphisms of the code form the automorphism group of the code.} 

The reasoning for this goes as follows. We may consider the divisor $D^{(i)}$ of $E_i$ which is formed as the sum of all the places of $E_i$ lying over the place $(z=1)$ in $E_0$; we denote these places as $(P_j)_{j=1}^{N^{(i)}}$. Naturally, we have $\deg(D^{(i)}) = N^{(i)}$. We may also consider some divisor $H^{(i)}$ of $E_i$ and the algebraic geometry code obtained by evaluating all functions in the Riemann-Roch space defined by $H^{(i)}$ at the rational places $P_j$:
\begin{equation}
    C \coloneq C_{\mathcal{L}}(D^{(i)},H^{(i)}) = \{(f(P_1), \ldots, f(P_{N^{(i)}}):f \in \mathcal{L}(H^{(i)})\}.
\end{equation}
We can attempt to define automorphisms of $C$ by taking any $\sigma \in \Gal(E_i/E_0)$ and letting
\begin{equation}
    \sigma(f(P_1), \ldots, f(P_{N^{(i)}})) \coloneq \left(f(\sigma(P_1)), \ldots, f(\sigma(P_{N^{(i)}}))\right).
\end{equation}
If $H^{(i)}$ is chosen to be invariant under $\Gal(E_i/E_0)$, meaning that $\sigma\left(H^{(i)}\right) = H^{(i)}$ for all $\sigma \in \Gal(E_i/E_0)$,\footnote{The definition of $\sigma (H^{(i)} )$ is the natural one. If $H^{(i)} = \sum_P\alpha_PP$, then $\sigma (H^{(i)} ) = \sum_P\alpha_P\sigma(P)$.} then one can show that the above defines an automorphism of the code $C$. Indeed, this may be established by first showing that if $\sigma(H^{(i)}) = H^{(i)}$, then $f \in \mathcal{L}(H^{(i)}) \implies \sigma(f) \in \mathcal{L}(H^{(i)})$ for all $\sigma \in \Gal(E_i/E_0)$. Moreover, the fact that $\Gal(E_i/E_0)$ acts transitively on the places in $E_i$ lying over the place $(z=1)$ of $E_0$ makes the code $C$ transitive.

\section{Transversally Addressable \texorpdfstring{$\mathsf{CCZ}$}{} Gates}
In this section, we will describe asymptotically good quantum codes over qudits of fixed dimension $q   =2^t = \Theta(1)$ with transversally addressable $\CCZ$ gates. 
When combined with the qudit-to-qubit transformations described in Section~6 of~\cite{he2025quantum} (which adapted methods from~\cite{golowich2024asymptoticallygoodquantumcodes,nguyen2024goodbinaryquantumcodes}), this will yield our main result, Theorem~\ref{thm:main_result_1}.
Here, it will suffice to show the following.

\begin{theorem}\label{thm:technical_AG_main}Let $q=r^2$ be a fixed power of two, where $r \geq 8$. There exists an asymptotically good family of qudit $\CSS$ codes $\mathcal{Q}_{\AG}$ for qudits of dimension $q$ with rate at least $\frac{1}{r(r-1)}$ and relative distance at least $\frac{1}{4}-\frac{3}{2}\frac{1}{r-1}-\frac{1}{r(r-1)}$. Given a member of the family with $k$ logical qudits, let $A,B,C \in [k]$ label any three logical qudits, and let $\gamma \in \mathbb{F}_q$. Then,
\begin{enumerate}
    \item On a single code block of $\mathcal{Q}_{\AG}$, $\overline{\CCZ^\gamma[A,B,C]}$ may be implemented by a depth-7 circuit of physical $\CCZ$ gates;
    \item On three distinct code blocks of $\mathcal{Q}_{\AG}$, $\overline{\CCZ^\gamma_{123}[A,B,C]}$ may be implemented by a depth-1 circuit of physical $\CCZ$ gates.
\end{enumerate}
\end{theorem}
\begin{remark}
    On two distinct code blocks of $\mathcal{Q}_{\AG}$, $\overline{\CCZ_{112}^\gamma[A,B,C]}$ may be implemented by a depth-$3$ circuit of physical $\CCZ$ gates. We do not show this explicitly as it is similar to the other two cases.
\end{remark}

\begin{remark}
    The implementation depth of each operation may be trivially reduced to one, i.e., a strongly transversal implementation, by duplicating qudits a constant number of times. This would produce another asymptotically good code, albeit with worse rate and relative distance.
\end{remark}
\begin{remark}
    We only consider the $\CCZ$ gate for brevity, although the result may be easily generalised to arbitrary diagonal gates (on a constant number of qudits) with entries $\pm 1$ on the diagonal, as is discussed in~\cite{he2025quantum}.
\end{remark}
\begin{remark}
    Under the qudit-to-qubit transformations described in~\cite{he2025quantum} (which in turn follow from~\cite{golowich2024asymptoticallygoodquantumcodes,nguyen2024goodbinaryquantumcodes}), we can use Theorem~\ref{thm:technical_AG_main} to construct an asymptotically good qubit code with addressable inter-block and intra-block addressable $\CCZ$ gates with strongly transversal (i.e. depth-$1$) implementations, thus proving Theorem~\ref{thm:main_result_1}. 
    This is possible because the field size is a constant here $(q = \Theta(1))$, in contrast to the result built from Reed-Solomon-based codes in~\cite{he2025quantum}. There, the growing qudit dimension leads to a qubit code which is only asymptotically good up to a polylogarithmic factor.
    For the sake of brevity, we choose not to restate the transformations in this paper and refer readers to Section~6 of~\cite{he2025quantum}.
\end{remark}
\begin{remark}
    Here, we have considered transversally addressing any \textit{single} logical $\mathsf{CCZ}$ gate on one, two, or three logical qubits. There is a stronger notion of addressability, as is considered in~\cite{guyot2025addressability}, for example, where one may ask about the possibility of being able to transversally address an arbitrary \textit{product} of $\mathsf{CCZ}$ gates on the logical qubits. In the Reed-Solomon-based construction of~\cite{he2025quantum}, we found that certain products of $\mathsf{CCZ}$ gates are transversally addressable, but not all are. One may ask the same question about the codes constructed here, where further information on the Galois group of the relevant field extension may reveal this possibility in the case of general algebraic geometry codes. We comment, however, that such a strong form of transversal addressability need not necessarily exist, and~\cite{guyot2025addressability} proves certain no-go theorems to this effect.
\end{remark}

\subsection{The Quantum Code} \label{sec:AG_qcode}

Recalling the preliminary material of Section~\ref{sec:further_AG_prelims}, especially the tower of function fields $\mathcal{E} = (E_0 \subseteq E_1 \subseteq E_2 \subseteq \ldots )$ considered in~\cite{stichtenoth2006transitive}, we are ready to describe our quantum code. Throuhgout, we fix some power of two, $q=r^2$, where $r \geq 8$. 
\subsubsection{The Underlying Algebraic Geometry Code} Section 4 of~\cite{stichtenoth2006transitive} considers classical algebraic geometry codes
\begin{equation}
    C_{a,b}^{(i)} \coloneq C_{\mathcal{L}}(D^{(i)}, aA^{(i)}+bB^{(i)}),
\end{equation}
where we recall that $D^{(i)}$ is the sum of the $N^{(i)}$ (rational) places of $E_i$ lying over the place $(z=1)$ in $E_0$. In addition, $A^{(i)}$ and $B^{(i)}$ are divisors of the function field $E_i$, which satisfy (see Theorem 1.7 in~\cite{stichtenoth2006transitive})
\begin{equation}\label{eq:principal_divisor}
    (w)^{E_i} = e_0^{(i)}\cdot A^{(i)} - e_\infty^{(i)}\cdot B^{(i)}.
\end{equation}
Here, $(w)^{E_i}$ is the principal divisor of $w$ in the function field $E_i$, where we recall that $E_0 = \mathbb{F}_q(z)$ and $z = w^{r-1}$. Moreover, $e_0^{(i)}$ and $e_\infty^{(i)}$ are positive integers satisfying
\begin{align}
    e_0^{(i)} &= r^{i-1}\cdot 2^{r(i)}\\
    e_\infty^{(i)} &= r^i\cdot 2^{s(i)},
\end{align}
for some non-negative integers $r(i)$ and $s(i)$.\footnote{For completeness, we remark that $e_0^{(i)}$ and $e_\infty^{(i)}$ are the \textit{ramification indices} of the places $(w=0)$ and $(w=\infty)$, respectively, in the extension $E_i/\mathbb{F}_q(w)$, although this terminology is not important to us.} In addition, in the definition $C_{a,b}^{(i)}$, $a$ and $b$ are integers satisfying
\begin{align}
    0 &\leq a \leq a_i\\
    0 &\leq b \leq b_i,
\end{align}
where (see the proof of Proposition 4.4 in~\cite{stichtenoth2006transitive})
\begin{align}
    a_i&=2e_0^{(i)}-2\\
    b_i&= (r-1)e_\infty^{(i)}-2.
\end{align}
For our purposes, we will exclusively choose
\begin{align}
    a &= \left\lfloor \frac{a_i}{4}\right\rfloor\\
    b &= \left\lfloor \frac{b_i}{4}\right\rfloor
\end{align}
and consider the corresponding classical codes $C_{a,b}^{(i)}$. 

\begin{proposition}
    The following hold for $C^{(i)}_{a,b}$.
    \begin{align}\label{eq:param_assump_1}
        a\cdot\deg A^{(i)} + b\cdot\deg B^{(i)} &< N^{(i)}\\
        (a+2)\deg A^{(i)} + (b+2)\deg B^{(i)} &> \frac{2N^{(i)}}{r-1}+\frac{N^{(i)}}{r(r-1)}.\label{eq:param_assump_2}
    \end{align} 
\end{proposition}

\begin{proof}
    Starting with the first equation, we have
\begin{align}
    a\cdot \deg A^{(i)} + b\deg B^{(i)} &\leq \frac{a_i}{4}\deg A^{(i)}+\frac{b_i}{4}\deg B^{(i)}\\
    &<\frac{1}{2}e_0^{(i)}\deg A^{(i)}+\frac{(r-1)}{4}e_\infty^{(i)}\deg B^{(i)}.
\end{align}
However, referring to Section 3 of~\cite{stichtenoth2006transitive} (in which the divisor $B^{(i)}$ is denoted as $G_0$), we find
\begin{equation}
    e_\infty^{(i)}\deg B^{(i)} = [E_i:\mathbb{F}_q(w)],
\end{equation}
but then, using Equation~\eqref{eq:principal_divisor}, and the fact that principal divisors have zero degree, we have also
\begin{equation}
    e_0^{(i)}\deg A^{(i)} = [E_i:\mathbb{F}_q(w)].
\end{equation}
Then, using the fact that $\mathbb{F}_q(w)/E_0$ is an extension of degree $r-1$, we have that $[E_i:\mathbb{F}_q(w)] = \frac{N^{(i)}}{r-1}$. Therefore, we have
\begin{equation}
    a\cdot\deg A^{(i)} + b\deg\cdot B^{(i)} < N^{(i)}\left(\frac{1}{2(r-1)}+\frac{1}{4}\right) < N^{(i)},
\end{equation}
thus establishing Equation~\eqref{eq:param_assump_1}. Next, we verify that we have Equation~\eqref{eq:param_assump_2}. Indeed, we have
\begin{align}
    (a+2)\deg A^{(i)} + (b+2)\deg B^{(i)} &> \left(\frac{a_i}{4}+1\right)\deg A^{(i)} + \left(\frac{b_i}{4}+1\right)\deg B^{(i)}\\
    &=\left(\frac{e_0^{(i)}+1}{2}\right)\deg A^{(i)} + \left(\frac{(r-1)e_\infty^{(i)}+2}{4}\right)\deg B^{(i)}\\
    &>\frac{N^{(i)}}{2(r-1)}+\frac{N^{(i)}}{4}\\
    &>\frac{2N^{(i)}}{r-1}+\frac{N^{(i)}}{r(r-1)},
\end{align}
where going into the third line we have again used $e_0^{(i)}\deg A^{(i)} = e_\infty^{(i)}\deg B^{(i)} = \frac{N^{(i)}}{r-1}$, and one may check that the last line holds for every $r \geq 8$.
\end{proof}

In the knowledge of Equations~\eqref{eq:param_assump_1} and~\eqref{eq:param_assump_2}, we may deduce the parameters of the code $C_{a,b}^{(i)}$ using the standard estimates (see Section 2.2 of~\cite{stichtenoth2009algebraic}).

\begin{lemma}
    The dimension of $C^{(i)}_{a,b}$ is 
    \begin{align}
        m^{(i)} \coloneq (a+1)\deg A^{(i)}+(b+1)\deg B^{(i)}-\frac{N^{(i)}}{r-1},
    \end{align}
    and its distance satisfies
    \begin{align}
        d\left(C_{a,b}^{(i)}\right) &\geq N^{(i)} - a\deg A^{(i)} - b\deg B^{(i)}.
    \end{align}
    The distance of the dual code $C_{a,b}^{(i)\perp}$ satisfies
    \begin{align}
    d\left(C_{a,b}^{(i)\perp}\right) &\geq(a+2)\deg A^{(i)}+(b+2)\deg B^{(i)}-2\frac{N^{(i)}}{r-1}.\label{eq:dual_distance_estiamte}
    \end{align}
\end{lemma}

\begin{proof}
    Equation~\eqref{eq:param_assump_1} simply says that
\begin{equation}\label{eq:AG_code_constraint_1}
    \deg\left(aA^{(i)}+bB^{(i)}\right) < N^{(i)}.
\end{equation}
Additionally, Equation~\eqref{eq:param_assump_2} implies that
\begin{equation}\label{eq:AG_code_constraint_2}
    \deg\left(aA^{(i)}+bB^{(i)}\right) > 2g(E_i)-2,
\end{equation}
where $g(E_i)$ is the genus of the function field $E_i$, which we find is
\begin{equation}
    g(E_i) = \frac{N^{(i)}}{r-1}+1-(\deg A^{(i)}+\deg B^{(i)})
\end{equation}
from Theorem 1.7 of~\cite{stichtenoth2006transitive}. Equations~\eqref{eq:AG_code_constraint_1} and~\eqref{eq:AG_code_constraint_2} allow us to determine the dimension and distance of the code $C_{a,b}^{(i)}$ as
\begin{align}
    m^{(i)} \coloneq \dim_{\mathbb{F}_q}\left(C_{a,b}^{(i)}\right) &= a\deg A^{(i)}+b\deg B^{(i)}+1-g(E_i)\\
    &=(a+1)\deg A^{(i)}+(b+1)\deg B^{(i)}-\frac{N^{(i)}}{r-1}\\
    d\left(C_{a,b}^{(i)}\right) &\geq N^{(i)} - a\deg A^{(i)} - b\deg B^{(i)},
\end{align}
respectively. As for the distance of the dual code. Equations~\eqref{eq:AG_code_constraint_1} and~\eqref{eq:AG_code_constraint_2} allow us to do estimate it as
\begin{align}
    d\left(C_{a,b}^{(i)\perp}\right) &\geq a\deg A^{(i)}+b\deg B^{(i)} -2g(E_i)+2\\
    &=(a+2)\deg A^{(i)}+(b+2)\deg B^{(i)}-2\frac{N^{(i)}}{r-1}.
\end{align}
This completes our proof.
\end{proof}

\subsubsection{Defining the Quantum Code}

Let us consider the places in the support of the divisor $D^{(i)}$, which we recall were the $N^{(i)}$ (rational) places of $E_i$ lying over the place $(z=1)$ in $E_0$. Referring to Section 2 of~\cite{stichtenoth2006transitive}, in particular point (F10), we find that in the extension $F_0/\mathbb{F}_q(z)$, the place $(z=1)$ of $\mathbb{F}_q(z)$ splits completely, and the places of $F_0 = \mathbb{F}_q(x_0)$ lying over it are the (rational) places $(x_0 = \eta)$, where $\eta \in \mathbb{F}_q$ satisfies $\eta^r + \eta \neq 0$. These $\eta$ are exactly the elements $\eta \in \mathbb{F}_q\setminus \mathbb{F}_r$, of which there are $q-r = r(r-1)$ and, indeed, note that the degree of the extension $F_0/\mathbb{F}_q(z)$ is $r(r-1)$. We are going to (arbitrarily) fix some $\eta' \in \mathbb{F}_q\setminus\mathbb{F}_r$ and partition the $N^{(i)}$ places of $E_i$ lying over the place $(z=1)$ into those that also lie over $(x_0 = \eta')$, and those that lie over $(x_0 = \eta)$ for some $\eta \neq \eta'$. We denote the former set as
\begin{equation}
    \bbeta = \{\beta_1, \ldots, \beta_{k^{(i)}}\}
\end{equation}
and the latter set as
\begin{equation}
    \bal = \{\alpha_1, \ldots, \alpha_{n^{(i)}}\}.
\end{equation}
Since the degree of the extension $F_0/\mathbb{F}_q(z)$ is $r(r-1)$, we have that the number of (rational) places of $E_i$ that lie over $(z=1)$ that also lie over $(x_0 = \eta')$ is
\begin{equation}
    k^{(i)} = |\bbeta| = \frac{N^{(i)}}{r(r-1)},
\end{equation}
and the number that lie over $(z=1)$ but lie over $(x_0 = \eta)$ for some $\eta \neq \eta'$ is 
\begin{equation}
    n^{(i)} = |\bal| = N^{(i)} - k^{(i)}.
\end{equation}

With this, let us write a generator matrix for the code $C_{a,b}^{(i)}$. This will be an $m^{(i)} \times N^{(i)}$ matrix whose columns will correspond to the places in the support of $D^{(i)}$. Moreover, we may let the first $k^{(i)}$ columns of the generator matrix correspond to the places $\bbeta$, and the latter $n^{(i)}$ columns of the generator matrix correspond to the places $\bal$. By performing row operations (which is equivalent to changing basis for the code), the generator matrix may be written into the form
\begin{equation}\label{eq:AG_gen_mat_form}
    \tilde{G} = \begin{pmatrix}
        I_{k^{(i)}}& G_1\\
        0 & G_0
    \end{pmatrix},
\end{equation}
where $I_{k^{(i)}}$ is the $k^{(i)} \times k^{(i)}$ identity matrix. This is possible because, combining Equations~\eqref{eq:param_assump_2} and~\eqref{eq:dual_distance_estiamte}, we have that $d\left(C^{(i)\perp}_{a,b}\right) > k^{(i)}$, which tells us that any $k^{(i)}$ columns of any generator matrix for $C_{a,b}^{(i)}$ must be linearly independent. Each row of $\tilde{G}$ is denoted as some $\tilde{g}^j$ for $j \in \left[m^{(i)}\right]$ and these are, by definition, the evaluations of some function in $\mathcal{L}(aA^{(i)}+bB^{(i)})$ at the places in the support of $D^{(i)}$. We may abuse notation and also denote this function as $\tilde{g}^j$.\footnote{Note that this is justified because the functions in $\mathcal{L}(aA^{(i)}+bB^{(i)})$ are in one-to-one correspondence with the codewords of $C_{a,b}^{(i)}$. To be explicit, since we have Equation~\eqref{eq:AG_code_constraint_1}, the evaluation map defining the algebraic geometry code is an isomorphism. Therefore, by identifying a particular codeword of $C_{a,b}^{(i)}$, in this case a row of $\tilde{G}$, we identify a unique function in $\mathcal{L}(aA^{(i)}+bB^{(i)})$.} We then let
\begin{equation}
    G = \begin{pmatrix}
        G_1\\G_0
    \end{pmatrix}
\end{equation}
be the restriction of $\tilde{G}$ to the columns corresponding to $\bal$, and we will denote the rows of $G$ as $g^j$ for $j \in \left[m^{(i)}\right]$. We will then define a quantum code in our usual way (see the preliminary material of~\cite{he2025quantum}), which is,
\begin{equation}
    \mathcal{Q}_{\AG} \coloneq \CSS(X,\mathcal{G}_0;Z,\mathcal{G}^\perp),
\end{equation}
where $\mathcal{G}_0$ and $\mathcal{G}$ are the vector spaces generated by the rows of $G_0$ and $G$, respectively, over $\mathbb{F}_q$. To establish the parameters of the quantum code $\mathcal{Q}_{AG}$, we need to establish a minimal assumption, denoted Assumption 2.9 in~\cite{he2025quantum}, which we repeat here for convenience.
\begin{assumption}\label{assump:independence}
    We assume that $\mathcal{G}_0 \cap \mathcal{G}_1 = 0$, and that the rows of $G_1$ are linearly independent.
\end{assumption}
We establish this in our case via the following claim.

\begin{claim}\label{claim:AG_code_two_mult}
    There is some fixed vector of non-zero entries $\underline{u} \in (\mathbb{F}_q^*)^{N^{(i)}}$ such that, for any two codewords $c^{(1)}, c^{(2)} \in C_{a,b}^{(i)}$, we have
    \begin{equation}
        \sum_{k=1}^{N^{(i)}}\underline{u}_kc^{(1)}_kc^{(2)}_k = 0.
    \end{equation}
\end{claim}
\begin{proof}
    Referring to Proposition 4.7 of~\cite{stichtenoth2006transitive}, we find that there is some fixed vector $\underline{u} \in (\mathbb{F}_q^*)^{N^{(i)}}$ such that
    \begin{equation}
        C_{a,b}^{(i)\perp} = \underline{u}\cdot C_{a_i-a,b_i-b}^{(i)},
    \end{equation}
    where the notation on the right-hand side denotes the code formed by multiplying every codeword in $C_{a_i-a,b_i-b}^{(i)}$ componentwise by the vector $\underline{u}$. Our particular choices of $a$ and $b$ imply that $a \leq a_i-a$ and $b \leq b_i-b$, and so we have that
    \begin{equation}
        \underline{u}\cdot C_{a,b}^{(i)} \subseteq \underline{u}\cdot C_{a_i-a,b_i-b}^{(i)} = C_{a,b}^{(i)\perp},
    \end{equation}
    and the result follows.
\end{proof}
Claim~\ref{claim:AG_code_two_mult} tells us that, for any $j_1, j_2 \in \left[m^{(i)}\right]$, we have
\begin{equation}
    \sum_{k=1}^{N^{(i)}}\underline{u}_k\tilde{g}^{j_1}_k\tilde{g}^{j_2}_k = 0
\end{equation}
and then, given the structure of $\tilde{G}$, we have the corresponding statement for the rows of $G$; for any $j_1, j_2 \in \left[m^{(i)}\right]$,
\begin{equation}\label{eq:AG_code_G_row_mult}
    \sum_{k=1}^{N^{(i)}}\underline{u}_kg_k^{j_1}g_k^{j_2} = \begin{cases}
        1&\text{ if } 1 \leq j_1 = j_2 \leq k^{(i)}\\
        0 &\text{ otherwise}
    \end{cases}
\end{equation}
(note that our field has characteristic two). Equation~\eqref{eq:AG_code_G_row_mult} can be straightforwardly used to establish Assumption~\ref{assump:independence}.

With this assumption established, we may bound the parameters of $\mathcal{Q}_{\AG}$. 

\begin{lemma}
    $\cQ_{\AG}$ has rate $\ge \frac{1}{r(r-1)}$ and relative distance $\ge \frac{1}{4}-\frac{3}{2(r-1)}-\frac{1}{r(r-1)} > 0$.
\end{lemma}

\begin{proof}
    The length of $\mathcal{Q}_{AG}$ is $n^{(i)} = N^{(i)} - k^{(i)}$, and its dimension is $k^{(i)}$ and so, in particular, its rate is at least $\frac{k^{(i)}}{N^{(i)}} = \frac{1}{r(r-1)}$. To lower bound the distance, we invoke equation~(45) of~\cite{he2025quantum} to obtain
    \begin{equation}
        d(\mathcal{Q}_{\AG}) \geq \min\left(d(C_{a,b}^{(i)}),d(C_{a,b}^{(i)\perp})\right)-k^{(i)},
    \end{equation}
    and using the previous estimates and parameter choices, it may be shown that $d(C_{a,b}^{(i)})\geq d(C_{a,b}^{(i)\perp})$, and that
    \begin{equation}
        d(\mathcal{Q}_{\AG}) \geq  d(C_{a,b}^{(i)\perp})-k^{(i)} \geq N^{(i)}\left[\frac{1}{4}- \frac{3}{2(r-1)}-\frac{1}{r(r-1)}\right],
    \end{equation}
    so that we establish relative distance at least $\frac{1}{4}-\frac{3}{2(r-1)}-\frac{1}{r(r-1)}$, which is positive for our choice $r \geq 8$.
\end{proof}

Having established the parameters of our quantum code, in particular that it has non-zero rate and relative distance, over a fixed field $\mathbb{F}_q$, we will establish the transversal addressability of the $\CCZ_q$ gate in the next two sections.

\subsection{Intra-Block \texorpdfstring{$\mathsf{CCZ}$}{} Gates}

We consider any member of the above code family indexed by $i$, which we recall has $n^{(i)}$ physical qudits and $k^{(i)}$ logical qudits. Given $A,B,C \in \left[k^{(i)}\right]$ and $\gamma \in \FF_q$, we will define a collection of $\CCZ$ gates that, when applied to a physical code state, will execute the operation $\ol{\CCZ^\gamma[A,B,C]}$. Note that, since the physical    qudits of the above code are in one-to-one correspondence with the set $\bal$, we may identify the set of physical qudits with this set. For example, a $\CCZ^\gamma$ gate on three physical qudits may be unambiguously written as $\CCZ^\gamma[\alpha_{k_1}, \alpha_{k_2}, \alpha_{k_3}]$.

Recall that the places $\bbeta$, which label our logical qudits, were chosen to lie over the place $(x_0 = \eta')$ of the function field $F_0$. It is a standard fact in Galois theory (see, for example, Appendix A.12 of~\cite{stichtenoth2009algebraic}) that given a Galois field extension $L/K$, and a further field $N$ with $L \supseteq N \supseteq K$, the extension $L/N$ is Galois. Therefore, given that the extensions $E_i/E_0$ are Galois, we deduce that the extensions $E_i/F_0$ are all Galois as well. Therefore, appealing to Theorem~\ref{thm:Galois_transitivity}, given any $A,B \in \left[k^{(i)}\right]$, there exists $\varphi_{AB} \in \Gal(E_i/F_0)$ such that
\begin{equation}
    \varphi_{AB}(\beta_A) = \beta_B.
\end{equation}
With this, we may define the logical $\CCZ$ operation in terms of physical $\CCZ$ gates, as follows.

\begin{theorem} \label{thm:ag-intra-block}
    Given the fixed vector $\underline{u}$ of non-zero entries from Claim~\ref{claim:AG_code_two_mult}, we call the vector formed by its first $k^{(i)}$ entries $\underline{x}$, and the vector formed by its latter $n^{(i)}$ entries $\underline{y}$.
    Let $g^A$ denotes the $A$-th row of $G$. 
    For any $\ket{\psi} \in \qcode_\AG$, it holds that 
    \begin{align}
        \ol{\CCZ^\gamma[A,B,C]} \ket{\psi} = \prod_{k=1}^{n^{(i)}} \CCZ^{\gamma\; \underline{x}_A^{-1}\underline{y}_k\;g^A_k}[\alpha_k,\varphi_{AB}(\alpha_k),\varphi_{AC}(\alpha_k)] \ket{\psi}.
    \end{align}
\end{theorem}

\begin{remark}
    We remark that the above action is well-defined, because given a place $\alpha_k$ and any $\varphi \in \Gal(E_i/F_0)$, we have that $\alpha_k$ lies over some place $(x_0 = \eta)$ of $F_0$ for $\eta \neq \eta'$. Recalling Lemma~\ref{lemma:place_group_action}, we have that $\varphi(\alpha_k)$ also lies over the place $(x_0 = \eta)$, and so in particular denotes another physical qudit.
    We encourage the reader to compare this theorem with Theorem~3.7 of~\cite{he2025quantum}, as the results and underlying ideas are similar.
\end{remark}
\begin{proof}
    Let us check the effect of applying $\prod_{k=1}^{n^{(i)}} \CCZ^{\gamma\; \underline{x}_A^{-1}\underline{y}_k\;g_k^A}[\alpha_k,\varphi_{AB}(\alpha_k),\varphi_{AC}(\alpha_k)]$ to a code state. By linearity, it suffices to check the effect of this gate on a logical computational basis state.

    To do so, we start by considering any $\mathbb{F}_q$-linear combination of the rows of $G$, $g^{(w)} = \sum_{j \in [m^{(i)}]}w_jg^j$. We may also consider the corresponding linear combination of rows of $\tilde{G}$, $\tilde{g}^{(w)} = \sum_{j \in [m^{(i)}]} w_j\tilde{g}^j$, which we recall corresponds to a (unique) element of the Riemann-Roch space $\mathcal{L}(aA^{(i)}+bB^{(i)})$, which we also denote as $\tilde{g}^{(w)}$. The state $\Ket{g^{(w)}}$ is simply the evaluations of the function $\tilde{g}^{(w)}$ at the places $\bal$.

    We have:
    \begin{align}
        & \prod_{k=1}^{n^{(i)}} \CCZ^{\gamma\; \underline{x}_A^{-1}\underline{y}_k\;g_k^A}[\alpha_k,\varphi_{AB} (\alpha_k),\varphi_{AC}(\alpha_k)] \ket{g^{(w)}} \\
        =~& \prod_{k=1}^{n^{(i)}} \exp\left( i\pi \tr\left( \gamma \;\underline{x}_A^{-1}\underline{y}_k\;g^A_k \cdot \tilde{g}^{(w)}(\alpha_k) \cdot \tilde{g}^{(w)}(\varphi_{AB}(\alpha_k)) \cdot \tilde{g}^{(w)}(\varphi_{AC}(\alpha_k)) \right) \right) \ket{g^{(w)}} \\
        =~& \exp \left( i\pi \tr \left( \gamma \cdot \underline{x}_A^{-1}\sum_{k=1}^{n^{(i)}} \underline{y}_k\;g^A_k \cdot \tilde{g}^{(w)}(\alpha_k) \cdot \tilde{g}^{(w)}(\varphi_{AB}(\alpha_k)) \cdot \tilde{g}^{(w)}(\varphi_{AC}(\alpha_k)) \right) \right) \ket{g^{(w)}}.\label{eq:mid_expression_interblock}
    \end{align}
    Note that $\tilde{g}^A, \tilde{g}^{(w)} \in \mathcal{L}(aA^{(i)}+bB^{(i)})$, and recall that $g^A_k$, which is the $k$-th element of the $A$-th row of $G$, is nothing more than the function $\tilde{g}^A$ evaluated at the place $\alpha_k$. Referring to Section 4 of~\cite{stichtenoth2006transitive}, we find that the divisors $A^{(i)}$ and $B^{(i)}$ are invariant under any $\sigma \in \Gal(E_i/E_0)$ (recall our discussion at the end of Section~\ref{sec:further_AG_prelims} for the meaning of invariance of a divisor under the action of an element of the Galois group). In particular, the divisor $aA^{(i)}+bB^{(i)}$ is invariant under the action of any element of the subgroup $\Gal(E_i/F_0)$ and so, given any $f \in \mathcal{L}(aA^{(i)}+bB^{(i)})$, we have $\sigma(f) \in \mathcal{L}(aA^{(i)}+bB^{(i)})$ for every $\sigma \in \Gal(E_i/F_0)$. Recalling Lemma~\ref{lemma:place_group_action}, we have that $\tilde{g}^{(w)}(\varphi_{AB}(\alpha_k))$, which is the function $\tilde{g}^{(w)}$ evaluated at the place $\varphi_{AB}(\alpha_k)$ is the same thing as the function $\varphi_{AB}^{-1}(\tilde{g}^{(w)})$ evaluated at the place $\alpha_k$. A similar comment applies to $\tilde{g}^{(w)}(\varphi_{AC}(\alpha_k))$; we may write this as
    \begin{align}
        \tilde{g}^{(w)}(\varphi_{AB}(\alpha_k)) &= \left(\varphi_{AB}^{-1}(\tilde{g}^{(w)})\right)(\alpha_k)\\
        \tilde{g}^{(w)}(\varphi_{AC}(\alpha_k)) &= \left(\varphi_{AC}^{-1}(\tilde{g}^{(w)})\right)(\alpha_k).
    \end{align}
    Moreover, we have that $\varphi_{AB}^{-1}(\tilde{g}^{(w)}), \varphi_{AB}^{-1}(\tilde{g}^{(w)}) \in \mathcal{L}(aA^{(i)}+bB^{(i)})$. From here, we will make use of the following claim.
    \begin{claim}\label{claim:AG_code_four_mult}
        Given any functions $f_1, f_2, f_3, f_4 \in \mathcal{L}(aA^{(i)}+bB^{(i)})$, we have that
        \begin{equation}
            \sum_{k=1}^{n^{(i)}}\underline{y}_kf_1(\alpha_k)f_2(\alpha_k)f_3(\alpha_k)f_4(\alpha_k) = \sum_{A'=1}^{k^{(i)}}\underline{x}_{A'}f_1(\beta_{A'})f_2(\beta_{A'})f_3(\beta_{A'})f_4(\beta_{A'}),
        \end{equation}
        where (as in the statement of Theorem~\ref{thm:ag-intra-block}, we have that $\underline{x}$ is formed from the first $k^{(i)}$ entries of the fixed all non-zeros vector $\underline{u} \in (\mathbb{F}_q^*)^{N^{(i)}}$ in Claim~\ref{claim:AG_code_two_mult}, and $\underline{y}$ is formed from its latter $n^{(i)}$ entries.
    \end{claim}
    \begin{proof}
        Given the functions $f_1, \ldots, f_4$, there are codewords in $C_{a,b}^{(i)}$ given by their evaluations at the places $\alpha_k$ and $\beta_A$. Turning again to Proposition 4.7 of~\cite{stichtenoth2006transitive}, we have that $(C_{a,b}^{(i)})^\perp = \underline{u}\cdot C_{a_i-a,b_i-b}^{(i)}$ for a fixed non-zero vector $\underline{u} \in (\mathbb{F}_q^*)^{N^{(i)}}$, which we also considered in Claim~\ref{claim:AG_code_two_mult}. We have
        \begin{equation}
            \underline{u}\cdot \left(C_{a,b}^{(i)}*C_{a,b}^{(i)}*C_{a,b}^{(i)}\right) \subseteq \underline{u}\cdot C_{3a,3b}^{(i)} \subseteq \underline{u}\cdot C_{a_i-a,b_i-b}^{(i)} = (C_{a,b}^{(i)})^\perp,
        \end{equation}
        where, given codes $C_1$ and $C_2$, $C_1*C_2$ denotes the set of vectors formed by componentwise multiplying any vector in $C_1$ with any vector in $C_2$. The first containment is true because, given any divisors $F$ and $G$, and functions $f \in \mathcal{L}(F)$ and $g \in \mathcal{L}(G)$, we have $fg \in \mathcal{L}(F+G)$. The second containment is true because we chose $a = \left\lfloor\frac{a_i}{4}\right\rfloor$ and $b = \left\lfloor \frac{b_i}{4}\right\rfloor$. It follows that
        \begin{equation}
            \sum_{A'=1}^{k^{(i)}}\underline{u}_{A'}f_1(\beta_{A'})f_2(\beta_{A'})f_3(\beta_{A'})f_4(\beta_{A'}) + \sum_{k=1}^{n^{(i)}}\underline{u}_{k+k^{(i)}}f_1(\alpha_k)f_2(\alpha_k)f_3(\alpha_k)f_4(\alpha_k)=0
        \end{equation}
        and the result follows.
    \end{proof}
    We consider the sum inside the exponent in Equation~\eqref{eq:mid_expression_interblock}. Using Claim~\ref{claim:AG_code_four_mult}, this is equal to
    \begin{align}
        &\sum_{k=1}^{n^{(i)}}\underline{y}_k\;g_k^A\cdot \tilde{g}^{(w)}(\alpha_k)\cdot \left(\varphi_{AB}^{-1}(\tilde{g}^{(w)})\right)(\alpha_k)\cdot \left(\varphi_{AC}^{-1}(\tilde{g}^{(w)})\right)(\alpha_k)\\ = &\sum_{A'=1}^{k^{(i)}}\underline{x}_{A'}\;\tilde{g}^A(\beta_{A'})\cdot \tilde{g}^{(w)}(\beta_{A'})\cdot\left(\varphi_{AB}^{-1}(\tilde{g}^{(w)})\right)(\beta_{A'})\cdot\left(\varphi_{AC}^{-1}(\tilde{g}^{(w)})\right)(\beta_{A'})\\
        =&\;\;\underline{x}_A\cdot \tilde{g}^{(w)}(\beta_A)\cdot\left(\varphi_{AB}^{-1}(\tilde{g}^{(w)})\right)(\beta_A)\cdot\left(\varphi_{AC}^{-1}(\tilde{g}^{(w)})\right)(\beta_A)\\
        =&\;\;\underline{x}_A\cdot \tilde{g}^{(w)}(\beta_A)\cdot \tilde{g}^{(w)}(\varphi_{AB}(\beta_A))\cdot \tilde{g}^{(w)}(\varphi_{AC}(\beta_A))\\
        =&\;\;\underline{x}_A\cdot \tilde{g}^{(w)}(\beta_A)\cdot\tilde{g}^{(w)}(\beta_B)\cdot\tilde{g}^{(w)}(\beta_C)\\
        =&\;\;\underline{x}_A\cdot w_Aw_Bw_C,
    \end{align}
    where going into the third line we have used that $\tilde{g}^A(\beta_{A'}) = \delta_{AA'}$ for $A,A' \leq k^{(i)}$ (recall the form of the generator matrix from Equation~\eqref{eq:AG_gen_mat_form}), and we have again used the form of this generator matrix going into the last line. We conclude that
    \begin{equation}\label{eq:AG_interblock_gw}
        \prod_{k=1}^{n^{(i)}}\CCZ^{\gamma\;\underline{x}_A^{-1}\underline{y}_k\;g_k^A}[\alpha_k,\varphi_{AB}(\alpha_k),\varphi_{AC}(\alpha_k)]\ket{g^{(w)}} = \exp\left(i\pi\tr\left(\gamma\; w_Aw_Bw_C\right)\right)\ket{g^{(w)}}.
    \end{equation}
    With this, we may consider any logical computational basis state. Indeed, given $w \in \mathbb{F}_q^k$, we have that
    \begin{equation}
        \ol{\ket{w}} \propto\sum_{g \in \mathcal{G}_0}\Ket{\sum_{A'=1}^{k^{(i)}}w_{A'}g^{A'}+g}
    \end{equation}
    where the normalisation is unimportant to us. By linearity, Equation~\eqref{eq:AG_interblock_gw} gives us that
    \begin{equation}
        \prod_{k=1}^{n^{(i)}}\CCZ^{\gamma\;\underline{x}_A^{-1}\underline{y}_k\;g_k^A}[\alpha_k,\varphi_{AB}(\alpha_k),\varphi_{AC}(\alpha_k)]\ol{\ket{w}} = \exp(i\pi\tr(\gamma \;w_Aw_Bw_C))\ol{\ket{w}},
    \end{equation}
    so that the gate on the left executes the logical gate $\ol{\CCZ^\gamma[A,B,C]}$, as required.
\end{proof}

\paragraph{Circuit Depth.} 
We have shown that we can apply a logical $\ol{\CCZ^\gamma[A,B,C]}$ gate by applying $n^{(i)}$ different $\CCZ$ gates on the physical qudits. We claim that one can arrange these $n^{(i)}$ different $\CCZ$ gates into a depth-$7$ circuit.

Note that for any physical qudit $\alpha_k$, there are exactly three physical $\CCZ$ gates that involve $\alpha_k$. This means that for any physical $\CCZ(\alpha_{k_1}, \alpha_{k_2}, \alpha_{k_3})$ gate, there are at most six other $\CCZ$ gates that involve any of $\alpha_{k_1}, \alpha_{k_2}$ or $\alpha_{k_3}$. 

We construct the seven layers of the circuit as follows. For the first layer, pick greedily as many of the $n^{(i)}$ $\CCZ$ gates that do not share any physical qudit. Then, all the remaining $\CCZ$ gates must share at least one physical qudit with one of the chosen gates in the first layer, otherwise they would have been included in the first layer. In particular, each remaining $\CCZ$ gate shares a physical qudit with at most five other remaining $\CCZ$ gates.

We repeat this process to select the $\CCZ$ gates in layer 2; we pick from the remaining $\CCZ$ gates as many as possible that do not share a physical qudit. At the end, all the remaining $\CCZ$ gates shares a physical qudit with at most four other remaining $\CCZ$ gates.

We continue this process four more times. After this sixth layer, all remaining $\CCZ$ gates will share no physical qudits with each other, and we can pick them all for our seventh layer.

\subsection{Inter-Block \texorpdfstring{$\mathsf{CCZ}$}{} Gates}

In the last section, we showed how to implement a logical $\CCZ$ gate on three qudits \emph{within} the same code state, meaning an intra-block $\CCZ$ gate. In this section, we will show how to implement a logical $\CCZ$ gate on any three logical qudits in three \emph{different} code blocks.

We will construct the logical operator that applies a $\CCZ$ gate between qudits $A$ in the first code state, $B$ in the second code state, and $C$ in the third code state. We will denote this logical operator by $\ol{\CCZ_{123}^\gamma[A, B, C]}$.  Likewise, a physical $\CCZ$ gate on qudits $\alpha_{k_1}$ in code state $1$, $\alpha_{k_2}$ in code state $2$, and $\alpha_{k_3}$ in code state $3$ will be denoted by $\CCZ_{123}[\alpha_{k_1}, \alpha_{k_2}, \alpha_{k_3}]$. Other notation will be carried over from the previous section, and the argument will be given more briefly as it is similar to the last section.

As we argued in the last section, given $A,B \in \left[k^{(i)}\right]$, there exists $\varphi_{AB} \in \Gal(E_i/F_0)$ such that
\begin{equation}
    \varphi_{AB}(\beta_A) = \beta_B
\end{equation}
by Theorem~\ref{thm:Galois_transitivity}, since the extension $E_i/F_0$ is Galois. Then, recalling the notation from Theorem~\ref{thm:ag-intra-block}, we have the following.

\begin{theorem} \label{thm:ag-inter-block}
    For any $\ket{\psi_1}, \ket{\psi_2}, \ket{\psi_3} \in \qcode_\AG$, it holds that 
    \begin{align}
        \ol{\CCZ_{123}^\gamma[A,B,C]} \ket{\psi_1} \ket{\psi_2} \ket{\psi_3} = \prod_{k=1}^{n^{(i)}} \CCZ_{123}^{\gamma \;\underline{x}_A^{-1}\underline{y}_k\;g^A_k}[\alpha_k,\varphi_{AB}(\alpha_k),\varphi_{AC}(\alpha_k)] \ket{\psi_1} \ket{\psi_2} \ket{\psi_3}.
    \end{align}
    Hence, we can define
    \begin{align}
        \ol{\CCZ_{123}^\gamma[A,B,C]} = \prod_{k=1}^{n^{(i)}} \CCZ_{123}^{\gamma \;\underline{x}_A^{-1}\underline{y}_k\;g^A_k}[\alpha_k,\varphi_{AB}(\alpha_k),\varphi_{AC}(\alpha_k)].
    \end{align}
\end{theorem}

\begin{proof}
    This proof will follow analogously to the proof of Theorem~\ref{thm:ag-intra-block}.  We will check the effect of applying $\prod_{k=1}^{n^{(i)}} \CCZ_{123}^{\gamma \;\underline{x}_A^{-1}\underline{y}_k\;g^A_k}[\alpha_k,\varphi_{AB}(\alpha_k),\varphi_{AC}(\alpha_k)]$ on $\ket{g^{(w)}}\ket{g^{(w')}}\ket{g^{(w'')}}$ where 
    \begin{align}
        g^{(w)} &= \sum_{j \in \left[m^{(i)}\right]}w_jg^j\\
        g^{(w')} &= \sum_{j \in \left[m^{(i)}\right]}w_j'g^j\\
        g^{(w'')} &= \sum_{j \in \left[m^{(i)}\right]}w_j''g^j,
    \end{align}
    which will be enough to check its behaviour on $\ol{\ket{w}}\ol{\ket{w'}}\ol{\ket{w''}}$, for any $w,w',w'' \in \mathbb{F}_q^{k^{(i)}}$, by linearity, which is enough to check its behaviour on any three codestates, again by linearity. Let
    \begin{align}
        \tilde{g}^{(w)} &= \sum_{j \in \left[m^{(i)}\right]}w_j\tilde{g}^j\\
        \tilde{g}^{(w')} &= \sum_{j \in \left[m^{(i)}\right]}w_j'\tilde{g}^j\\
        \tilde{g}^{(w'')} &= \sum_{j \in \left[m^{(i)}\right]}w_j''\tilde{g}^j,
    \end{align}
    be the corresponding functions in $\mathcal{L}(aA^{(i)}+bB^{(i)})$.

    We have that 
    \begin{align}
        & \prod_{k=1}^{n^{(i)}} \CCZ_{123}^{\gamma\; \underline{x}_A^{-1}\underline{y}_k\;g^A_k}[\alpha_k,\varphi_{AB} (\alpha_k),\varphi_{AC}(\alpha_k)] \ket{g^{(w)}}\ket{g^{(w')}}\ket{g^{(w'')}}\\
        =~& \prod_{k=1}^{n^{(i)}} \exp\left( i\pi \tr\left( \gamma\;\underline{x}_A^{-1}\underline{y}_k\; g^A_k \cdot \tilde{g}^{(w)}(\alpha_k) \cdot \tilde{g}^{(w')}(\varphi_{AB}(\alpha_k)) \cdot \tilde{g}^{(w'')}(\varphi_{AC}(\alpha_k)) \right) \right) \ket{g^{(w)}}\ket{g^{(w')}}\ket{g^{(w'')}} \\
        =~& \exp \left( i\pi \tr \left( \gamma\; \underline{x}_A^{-1}\cdot \sum_{k=1}^{n^{(i)}}\underline{y}_k\; g^A_k \cdot \tilde{g}^{(w)}(\alpha_k) \cdot \tilde{g}^{(w')}(\varphi_{AB}(\alpha_k)) \cdot \tilde{g}^{(w'')}(\varphi_{AC}(\alpha_k)) \right) \right) \ket{g^{(w)}}\ket{g^{(w')}}\ket{g^{(w'')}}.
    \end{align}
    The sum may be treated in essentially the same way as in the proof of Theorem~\ref{thm:ag-intra-block}. Indeed, we have that $\tilde{g}^A, \tilde{g}^{(w)}, \varphi_{AB}^{-1}(\tilde{g}^{(w')})$ and $\varphi_{AC}^{-1}(\tilde{g}^{(w'')})$ are all elements of $\mathcal{L}(aA^{(i)}+bB^{(i)})$, and so Claim~\ref{claim:AG_code_four_mult} applies, and using essentially the same manipulations we obtain
    \begin{align}
        &\sum_{k=1}^{n^{(i)}} \underline{y}_k\;g^A_k \cdot \tilde{g}^{(w)}(\alpha_k) \cdot \tilde{g}^{(w')}(\varphi_{AB}(\alpha_k)) \cdot \tilde{g}^{(w'')}(\varphi_{AC}(\alpha_k)) \\
        =\;&\sum_{k=1}^{n^{(i)}} \underline{y}_k\;g^A_k \cdot \tilde{g}^{(w)}(\alpha_k) \cdot \left(\varphi_{AB}^{-1}(\tilde{g}^{(w')})\right)(\alpha_k) \cdot \left(\varphi_{AC}^{-1}(\tilde{g}^{(w'')})\right)(\alpha_k) \\
        =\;&\sum_{A'=1}^{k^{(i)}} \underline{x}_{A'}\;\tilde{g}^A(\beta_{A'}) \cdot \tilde{g}^{(w)}(\beta_{A'}) \cdot \left(\varphi_{AB}^{-1}(\tilde{g}^{(w')})\right)(\beta_{A'}) \cdot \left(\varphi_{AC}^{-1}(\tilde{g}^{(w'')})\right)(\beta_{A'}) \\
        =\; &\underline{x}_A\cdot \tilde{g}^{(w)}(\beta_A)\cdot \tilde{g}^{(w')}(\varphi_{AB}(\beta_A))\cdot \tilde{g}^{(w'')}(\varphi_{AC}(\beta_A))\\
        =\; &\underline{x}_A\cdot \tilde{g}^{(w)}(\beta_A)\cdot \tilde{g}^{(w')}(\beta_B)\cdot \tilde{g}^{(w'')}(\beta_C)\\
        =\; &\underline{x}_A\cdot w_Aw_B'w_C'',
    \end{align}
    so 
    \begin{align}
        \prod_{k=1}^{n^{(i)}} \CCZ_{123}^{\gamma\; \underline{x}_A^{-1}\underline{y}_k\;g^A_k}[\alpha_k,\varphi_{AB}(\alpha_k),\varphi_{AC}(\alpha_k)] \ket{g^{(w)}}\ket{g^{(w')}}\ket{g^{(w'')}}
        &= \exp({i\pi\tr(\gamma\; w_Aw_B'w_C'')}) \ket{g^{(w)}}\ket{g^{(w')}}\ket{g^{(w'')}},
    \end{align}
    as required.
\end{proof}

\paragraph{Circuit Depth.}
The gate 
\begin{align}
    \ol{\CCZ_{123}^\gamma[A,B,C]} = \prod_{k=1}^{n^{(i)}} \CCZ_{123}^{\gamma\; \underline{x}_A^{-1}\underline{y}_k\;g^A_k}[\alpha_k,\varphi_{AB}(\alpha_k),\varphi_{AC}(\alpha_k)]
\end{align}
can be implemented in depth $1$, since every physical qudit appears in exactly one physical $\CCZ$ gate.


\ifsub
\else

\section{Acknowledgements.}

We are grateful to Venkatesan Guruswami for general conversations on algebraic geometry codes and to Swastik Kopparty for pointing us to transitive algebraic geometry codes.
\fi 

\bibliographystyle{alpha}
\bibliography{main}

\appendix

\end{document}